\documentclass[10pt, conference]{IEEEtran}
\IEEEoverridecommandlockouts
\usepackage{epsfig, amsmath,amssymb,epsf,cite,subfigure,scalefnt}

\def\cA{{\mathcal{A}}}  \def\cC{{\mathcal{C}}} 
   
  \def\cK{{\mathcal{K}}} \def\cL{{\mathcal{L}}}
 \def\cN{{\mathcal{N}}}  
   
  \def\cW{{\mathcal{W}}} 
 \def\cZ{{\mathcal{Z}}}

\def\ba{{\mathbf{a}}} \def\bb{{\mathbf{b}}}   
\def\bf{{\mathbf{f}}}    
    
   \def\bs{{\mathbf{s}}} 
   \def\bx{{\mathbf{x}}} \def\by{{\mathbf{y}}}
\def\bz{{\mathbf{z}}}

\def\bA{{\mathbf{A}}} \def\bB{{\mathbf{B}}}   
 \def\bG{{\mathbf{G}}} \def\bH{{\mathbf{H}}} \def\bI{{\mathbf{I}}} 
   \def\bN{{\mathbf{N}}} 
\def\bP{{\mathbf{P}}}    
  \def\bW{{\mathbf{W}}}

\def\det{\mathop{\mathrm{det}}}

\def\tr{\mathop{\mathrm{tr}}}

\def\dim{\mathop{\mathrm{dim}}}

\def\rank{\mathop{\mathrm{rank}}}
\def\span{\mathop{\mathrm{span}}}

     \def\d4{\!\!\!\!}

 \def\bsh{\backslash}  
\def\bPi{\mathbf{\Pi}}

   \def\C{{\mathbb{C}}} 

\def\bzero{\mathbf{0}}

\def\lp{\left(}     \def\rp{\right)}    \def\ls{\left\{}    \def\rs{\right\}}    \def\lS{ \left[ }
\def\rS{ \right] }                

  \def\barm{\bar{m}}

    \def\bbA{\bar{\bA}}
\def\bbB{\bar{\bB}}    
        
    \def\tbA{\widetilde{\bA}} \def\tbB{\widetilde{\bB}}
 \def\tbz{\tilde{\bz}}  \def\tbs{\tilde{\bs}}
    \def\hbB{\widehat{\bB}}
  \def\-{\! - \!}  \def\+{\! + \!}  \def\={\! = \!}  \def\>{\! > \!}

\newtheorem{theorem}{Theorem}
\newtheorem{lemma}{Lemma}

\newcommand{\bet}{\begin{table}}
\newcommand{\eet}{\end{table}}
\newcommand{\btt}{\begin{tabular}}
\newcommand{\ett}{\end{tabular}}
\newcommand{\bec}{\begin{center}}
\newcommand{\eec}{\end{center}}
\newcommand{\bef}{\begin{figure}}
\newcommand{\eef}{\end{figure}}
\newcommand{\beq}{\begin{eqnarray}}
\newcommand{\eeq}{\end{eqnarray}}
\newcommand{\bea}{\begin{array}}
\newcommand{\eea}{\end{array}}

\newenvironment{proof}[1][Proof]{\begin{trivlist}
\item[\hskip \labelsep {\bfseries #1}]}{\end{trivlist}}

\newcommand{\qed}{\nobreak \ifvmode \relax \else
\ifdim\lastskip<1.5em \hskip-\lastskip
\hskip1.5em plus0em minus0.5em \fi \nobreak
\vrule height0.75em width0.5em depth0.25em\fi}

\def\val{1.0}

\begin{document}

\title{ Spatial Degrees of Freedom of the Multicell MIMO Multiple Access Channel}
\author{Taejoon~Kim$^\dag$,
        David~J.~Love$^\dag$,
        Bruno~Clerckx$^\ddag$, 
        and~Duckdong~Hwang$^\ddag$ \\
\IEEEauthorblockA{ \small$^\dag$School of Electrical and Computer Engineering, Purdue University, West Lafayette, IN 47906, USA\\
Email: kim487@ecn.purdue.edu, djlove@ecn.purdue.edu \\
$^\ddag$Samsung Electronics, Yongin-Si, Gyeonggi-Do, Korea 446-712\\
Email: bruno.clerckx@samsung.com, duckdong.hwang@samsung.com}
\thanks{This work was supported in part by Samsung Electronics.}}
\maketitle

\begin{abstract}
We consider a homogeneous multiple cellular scenario with multiple users per cell, i.e.,  
$K\geq 1$ where $K$ denotes the number of users in a cell. In this scenario, 
a degrees of freedom outer bound as well as an achievable scheme that attains the degrees of freedom outer bound
of the multicell multiple access channel (MAC) with constant channel coefficients are investigated.  
The users have $M$ antennas, and the base stations are equipped with $N$ antennas. 
The found outer bound is general in that it characterizes a degrees of freedom upper bound for $K\geq 1$
and $L>1$ where $L$ denotes the number of cells. 
The achievability of the degrees of freedom outer bound is studied for two cell case (i.e., $L=2$).
The achievable schemes that attains the degrees of freedom outer bound for $L\=2$ are based on two approaches. 
The first scheme is a simple zero forcing with $M\=K\beta\+\beta$ and $N\=K\beta$, and 
the second approach is null space interference alignment with $M\=K\beta$ and $N\=K\beta\+\beta$ 
where $\beta>0$ is a positive integer.  
\end{abstract}
\IEEEpeerreviewmaketitle

\def\val{1.0}
\scalefont{\val}
%
%
\section{Introduction}


Challenges in identifying the exact information-theoretic capacity of general interfering networks 
motivates people to study the approximated capacity in the high SNR 
regime (some of which can be practically achieved in small cell scenarios) by analyzing the number 
of resolvable signal dimensions in terms of the degrees of freedom of the network. 
Initial works include the degrees of freedom and/or capacity region characterization for the MIMO 
multiple access channel (MAC) \cite{Tse1} and MIMO broadcast channel \cite{vishwanath,viswanath,weingarten}. 
Recently, the degrees of freedom have been studied broadly for various kinds of 
networks \cite{Maddah, Jafar2, Jafar1, Cadambe1, Cadambe2, Gou, Suh, Suh1, Honig1}. 
The key innovation used to prove the achievability of the degrees of freedom in 
\cite{Jafar1, Cadambe1, Cadambe2, Gou, Suh} is interference alignment.              
Interference alignment generates overlapping interference subspaces at the receiver while 
keeping the desired signal spaces distinct. 
When the degrees of freedom outer bound is achieved by some scheme, we say the scheme obtains the \textit{optimal degrees of freedom}.



Interference alignment in a time (or frequency) 
varying channel with finite or infinite symbol extension is the main focus of the work in \cite{Cadambe1, Cadambe2, Gou, Suh}.
For instance, interference alignment achieves the optimal degrees of freedom 
for the $K$ by $L\=2$ (or $K\=2$ by $L$) single antenna user X network with finite symbol extension 
\cite{Cadambe1}. For X networks with $K>2$ and $L>2$, interference alignment requires 
infinite symbol extension in order to be close to the outer bound \cite{Cadambe1}. 
In the case of constant channel coefficients, 
the spatial degrees of freedom have been investigated in \cite{Jafar2, Maddah, Jafar1, Cadambe2, Suh1, Honig1}. 
The optimal degrees of freedom of the two by two MIMO X channel has the optimal degrees of freedom of $\frac{4}{3}M$ when each 
node has $M>1$ antennas \cite{Jafar2, Maddah}.
With $M$ antennas at each transmitter and $N$ antennas at each receiver, 
Ref. \cite{Jafar1} characterizes the optimal degrees of freedom for the two user interference channel. 
Remarkably, simple zero forcing is sufficient to achieve the optimal degrees of freedom \cite{Jafar2, Jafar1}. 
The interference alignment in a three-user interference channel with $M$ antennas at each node yields the optimal 
degrees of freedom of $\frac{3M}{2}$ when $M$ is even (when $M$ is odd a 
two symbol extension is required to achieve $\frac{3M}{2}$) \cite{Cadambe2}. 
An achievable scheme where each user can obtain one degree of freedom for 
two cell network with a constant channel coefficient is the main focus of \cite{Suh1}.      
Necessary antenna dimension conditions for a linear scheme to provide one degree of freedom per user 
are formulated in terms of the number of users and the number of cells in \cite{Honig1}.
The general characterization of the optimal degrees of freedom for MIMO networks 
with constant channel coefficients still remains unknown. 

In this paper, we study the degrees of freedom for the $L$-cell and $K$-user MIMO MAC where 
the network consists of $L>1$ homogenous cells with $K\geq 1$ users per cell.  
Spatial resources are mainly utilized with constant channel coefficients to study the degrees of freedom. 
So, we do not consider symbol extension to utilize time or frequency resources.  
We first provide a degrees of freedom outer bound for the $L$-cell and $K$-user MIMO MAC. 
Then, two schemes that achieve the degrees of freedom outer bound are constructed for $L=2$, i.e., two-cell case.  
The first scheme is a simple transmit zero forcing with $N\=K\beta$ and $M\=K\beta\+\beta$ and 
the second one is a \emph{null space interference alignment} with $N\=K\beta\+\beta$ and $M\=K\beta$, 
where $\beta$ is a positive integer. 
The \emph{optimal degrees of freedom} for two-cell MIMO MAC is shown to be 
$2K\beta$, when $M\=K\beta$ and $N\=K\beta\+\beta$ 
or $M\=K\beta\+\beta$ and $M\=K\beta$.

The keys to the degrees of freedom outer bound are to construct a subset network of the $L$-cell and $K$-user MIMO MAC and
to allow full cooperation between users and their corresponding basestations in a certain manner. 
When $N>M$ (deplorable uplink scenario), the achievable scheme is based on null space interference alignment. 
Null space interference alignment relies on each base 
station using a carefully chosen null space plane to project the out-of-cell interference to a lower dimensional
space than its original dimension so that the null space plane can jointly mitigate the 
degrees of freedom loss. 
The converse and achievability lead to the optimal degrees 
of freedom characterization for the two cell case. 
Notice that by the uplink and downlink duality, the uplink scenario is converted to the 
downlink scenario as shown in \cite{Suh1, Honig1}. Thus, without loss of generality, the degrees of freedom results 
in this paper are also applicable to the downlink scenario. 

The organization of the paper is as follows. Section \ref{section2} describes 
the system model for the $L$-cell and $K$-user MIMO MAC. In Section \ref{section3}, we derive a  
degrees of freedom outer bound for the multicell MIMO MAC when $K\!\geq \!1$ and $L\!>\!1$. 
Studying the achievability and optimal degrees of freedom for the two-cell MIMO MAC is in Section 
\ref{section4}. The paper is concluded in Section \ref{section_conlusions}.

\section{$L$-cell and $K$-user MIMO MAC} \label{section2}

The network consists of $L$ homogeneous cells. In each cell there are $K\!\geq\! 1$ users 
and one base station where the user (transmitter) has $M\!\geq\! 1$ antennas and the base station
(receiver) is equipped with $N\!\geq\! 1$ antennas. We introduce an index $\ell k$ to denote 
the user $k$ in the cell $\ell$ for $\ell\in\cL$ and $k\in\cK$ where
$\cL=\{1,\ldots, L\}$ and $\cK=\{1,\ldots, K\}$,
respectively. In the $L$-cell and $K$-user MIMO MAC, a total of $LK$ users simultaneously transmit data to destined base stations. 
For instance, a three-cell and two-user MIMO MAC is shown in Fig. \ref{Fig1}. 
Here, user indices $\{ \ell 1, \ell 2 \}$  denote users in cell $\ell$.    
The input-output relation of the channel at the $t$th discrete time slot is described by 
\beq
\by_{m}(t)= \sum_{\ell=1}^{L}\sum_{k=1}^{K}\bH_{m,\ell k}\bx_{\ell k}(t)+\bz_{m}(t), \ \forall m \in \cL \label{channel_model1}
\eeq
where $\by_{m}(t)\!\in\!\C^{N\times 1}$ and $\bz_{m}(t)\!\in\! \C^{N\times 1}$ denote the received signal 
vector and additive noise vector at the base station $m$, respectively.
Each entry of $\bz_{m}(t)$ is independent and identically distributed (i.i.d.) with $\cC\cN(0,1)$.
The vectors $\bx_{\ell k}(t)\!\in\!\C^{M\times 1}$ represents the channel input at user $\ell k$. 
The $\bx_{\ell k}(t)$ is subject to an average power constraint
\beq
\tr\lp E\lS \bx_{\ell k}(t)\bx_{\ell k}^*(t) \rS \rp \leq \rho, \ \forall k\in\cK, \forall \ell\in\cL \label{power_const1}
\eeq
where $\rho$ represents SNR.  
The matrix $\bH_{m,\ell k}\!\in\!\C^{N\times M}$ denotes the channel with constant coefficients from user $\ell k$
to base station $m$. 
In \eqref{channel_model1}, the matrices $\{\bH_{m,mk}\}_{k\in\cK}$ represent the desired signal 
channel at base station $m$ while the matrices $\{\bH_{m,\ell k}\}_{\ell \in\cL\bsh m, k\in\cK}$ 
carry out-of-cell interference to base station $m$. 
The channel matrices are realized from i.i.d. and continuous distribution such that each entry 
is i.i.d. and the distribution of each entry has compact support.
This channel model almost surely ensures all channel matrices are \emph{nondegenerate}, i.e., 
$\rank(\bH_{m,\ell k})\=\min(M,N)$\footnote{The $\rank(\bA)$ for 
$\bA\in\C^{N \times M}$ is defined as $\rank(\bA)=\dim(ran(\bA))$
where $ran(\bA)=\{ \by\in\C^{N\times 1}: \by=\bA\bx, \forall\bx\in\C^{M\times 1} \}$ 
and $\dim(\cA)$ extracts the number of basis of the subspace $\cA$. 
Null space of $\bA$ is defined by $null(\bA)=\{ \ba\in\C^{M\times 1}:\mathbf{0}=\bA\ba \}$.}
and the event for $(\bH_{m,\ell k})_{i,j}\=\infty$ is negligible. 
Throughout the paper we assume perfect channel knowledge of all links at all nodes. 

Define $W_{\ell k}(\rho)$ as a message from user $\ell k$ to the destined base station $\ell$. 
The message $W_{\ell k}(\rho)$ is uniformly distributed in a $(n,2^{nR_{\ell k}(\rho)})$ codebook 
$\cZ(\rho)$ and messages at different users are independent each other.
The message $W_{\ell k}(\rho)$ is mapped to $\bx_{\ell k}$ in \eqref{channel_model1}. 
Then, the information transfer rate $R_{\ell k}(\rho)$ of message $W_{\ell k}(\rho)$ is 
said to be achievable if the rate of decoding error can be made arbitrarily small by choosing 
appropriate channel block length $n$. The capacity region $\cC(\rho)$ is defined as the convex closure of all 
achievable rate tuples $\{R_{\ell k}(\rho)\}_{\ell \in\cL, k\in\cK}$.    
We define spatial degrees of freedom of multicell MIMO MAC as
\beq
\Lambda_d \d4&=&\d4 \lim_{\rho\rightarrow\infty} 
                                     \sum\limits_{ \ls R_{\ell k}(\rho) \rs_{\ell\in\cL, k\in\cK}\in\cC(\rho)}  
                                      \frac{R_{lk}(\rho)}{\log(\rho)}. \label{definition_dof}
\eeq
The expression in \eqref{definition_dof} approximates the capacity region when the available power  
$\rho$ is arbitrary large. 
In the absence of exact knowledge of the capacity region, the degrees of freedom provides insight into network MIMO performance trends.
For the sake of simplicity, in what follows, we omit the $\rho$ attached to $W_{\ell k}(\rho)$ and $R_{\ell k}(\rho)$.
In addition, with an abuse of notation, $\by_m(t)$, $\bz_m(t)$, and $\bx_{\ell k}(t)$ in \eqref{channel_model1} 
are simplified to $\by_m$, $\bz_m$, and $\bx_{\ell k}$. 


\begin{figure}[t]
\centering
\includegraphics[width=6.0cm, height=6.0cm]{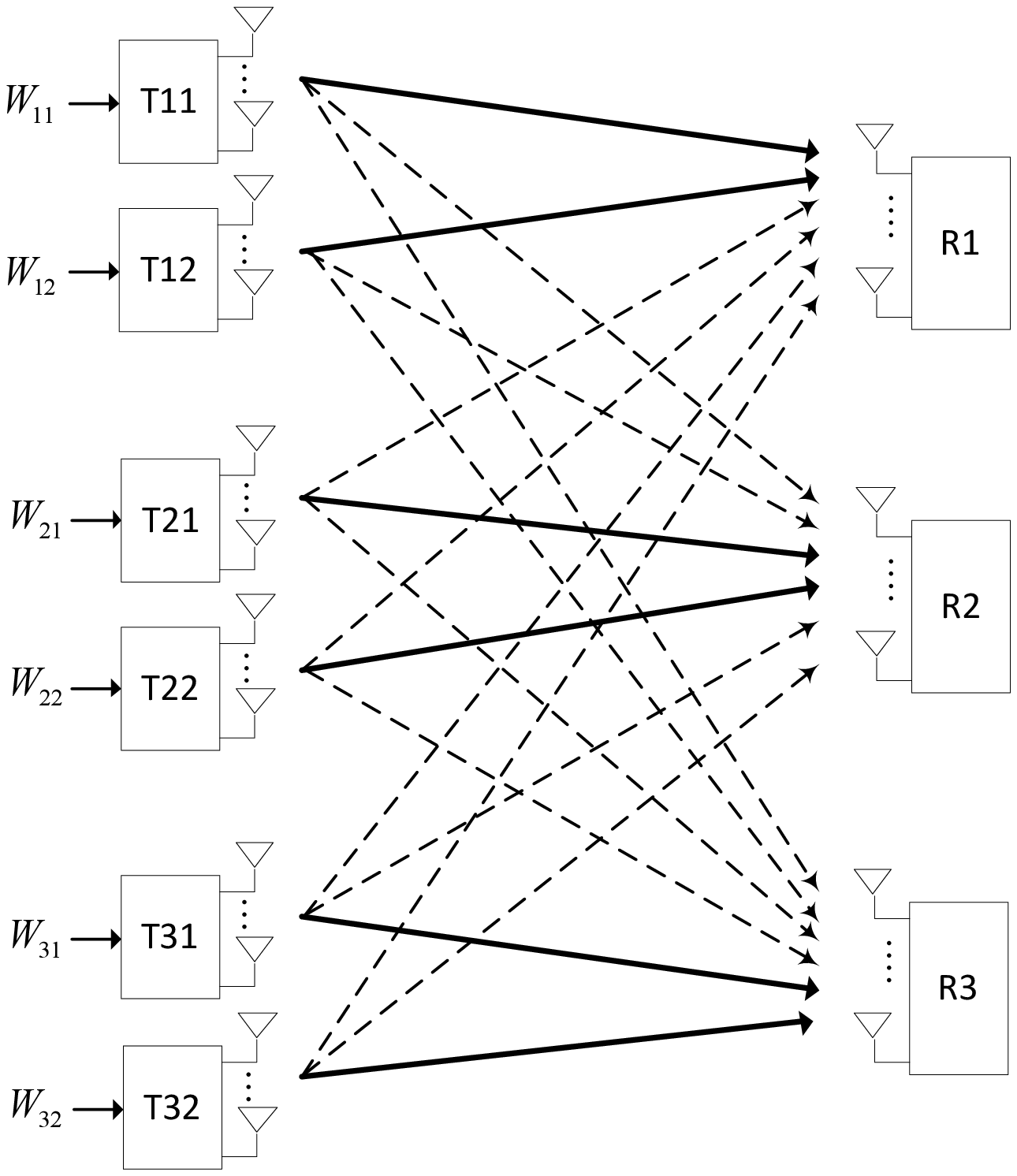}
\caption{Multicell MIMO MAC with $L=3$ and $K=2$.}
\label{Fig1}
\end{figure}

\section{Degrees of Freedom Outer Bound of the $L$-cell and $K$-user MIMO MAC} \label{section3}

A degrees of freedom outer bound for the $L$-cell and $K$-user MIMO MAC where the transmitter and receiver have
$M$ and $N$ antennas, respectively, is characterized as follows.  

\begin{theorem} \label{theorem_MC_MAC_outer_bound}
The degrees of freedom of the $L$-cell and $K$-user MIMO MAC with $L>1$ and $K\geq 1$, whose channel matrices are nondegenerate, 
is bounded by
\beq
\Lambda_d \!\leq\! \min\lp KLM, LN, \lambda_d\rp . \label{theorem_MC_MAC_outer_bound_equation}
\eeq 
where
\beq
\lambda_d \= KL\min\!\lp\! \frac{\max(KM,(L\-1)N)}{K\+L\-1},  \frac{\max((L\-1)M,N)}{K\+L\-1} \!\rp \nonumber 
\eeq
\end{theorem} 
\begin{proof}
A trivial outer bound is obtained by allowing perfect cooperation among $KL$ users and their corresponding $L$ basestations 
of the $L$-cell and $K$-user MIMO MAC as
\beq
\Lambda_d \leq \min\lp KLM, NL \rp. \label{3.24}
\eeq

The main ingredient to formulate the outer bound in \eqref{theorem_MC_MAC_outer_bound_equation} is to split the whole 
message set $\cW\=\ls W_{\ell k} \rs_{\ell\in\cL, k\in\cK}$ into smaller subsets, characterize the degrees of freedom 
outer bound associated with this small subset, and combine all degrees of freedom characterizations associated with 
all of the subsets to compute \eqref{theorem_MC_MAC_outer_bound_equation}. 

First, we define a network which is a subset of $L$-cell and $K$-user MIMO MAC.
The subset network is defined as a $L$-cell heterogeneous MIMO uplink channel, where $L-1$ cells (among $L$ cells) form 
the $(L-1)$-user MIMO interference channel and single cell forms the $K$-user MIMO MAC. 
We refer to this network as the $(L\-1,1)$ \emph{uplink HetNet}. 
Fig. \ref{Fig6} represents $(2,1)$ uplink HetNet where cell $1$ is a $2$-user MIMO MAC and a cell $2$ and cell $3$ constitute
$2$-user MIMO interference channels. 
The $(L\-1,1)$ uplink HetNet is formed by designating the $\ell$th cell (among $L$ cells) as the $K$-user MIMO MAC.
Then, the other $L\-1$ cells in $\cL\bsh\ell$ form $(L\-1)$-user MIMO interference channels by selecting the $k$th user in 
each of the cells in $\cL\bsh\ell$, i.e., the index set for the $L\-1$ users is $\ls 1k, \ldots, \ell\-1 \ k, \ell\+1 \ k, \ldots, Lk  \rs$.  

The message set corresponding to the $K$-user MIMO MAC is $\ls W_{\ell q} \rs_{q\in\cK}$.  
The message set associated with $(L\-1)$-user MIMO interference channel is given by 
$\ls W_{pk} \rs_{k\in\cL\bsh\ell}$. 
Then, the messages set of $(L\-1,1)$ HetNet is defined by 
\beq
\cW^{\ell k} = \ls W_{\ell q} \rs_{q\in\cK} \cup \ls W_{pk} \rs_{k\in\cL\bsh\ell}. \label{message_spliting}
\eeq
The degrees of freedom outer bound is first argued for each of the $LK$ sets $\ls \cW^{\ell k} \rs_{\ell\in\cL, k\in\cK}$, and
$KL$ outer bounds are combined by accounting for overlapped messages.     

Now allow perfect cooperation between $L\-1$ users and corresponding $L\-1$ receivers of the $(L\-1)$-user MIMO interference channel. 
Then, if we assume perfect cooperation between $K$ users in cell $\ell$, 
the $(L-1,1)$ uplink HetNet with $\cW^{\ell k}$ becomes two-user interference channel,  
where the first link has the transmit and receive antenna pair $\lp KM, N \rp$ and the second link 
consists of $\lp  (L\-1)M,  (L\-1)N \rp$ transmit and receive antenna pair. 
The optimal spatial degrees of freedom of the $(M_1, N_1)$, $(M_2, N_2)$ two-user MIMO interference channel
is characterized by $\min( M_1+M_2, N_1+N_2, \max(M_1, N_2), \max(M_2, N_1) )$ in \cite{Jafar2}.    
Thus, by utilizing this result in \cite{Jafar2}, the degrees of freedom outer bound associated with message set 
$\cW^{\ell k}$ is given by 
\beq
\min \d4&(&\d4 (K\+L\-1)M, LN,  \nonumber \\ 
\d4& &\d4 \hspace{0.5cm} \max\lp KM, (L\-1)N \rp, \max\lp (L\-1)M,N \rp  ). \label{3.15}
\eeq

Since the bound in \eqref{3.15} does not alter for the message set $\cW^{\bar{\ell} \bar{k}}$ with $\bar{\ell}\neq\ell$ and 
$\bar{k}\neq k$, the degrees of freedom outer bound for the other message set $\big\{ \cW^{\bar{\ell} \bar{k}} \big\}_{\bar{\ell}\neq\ell, \bar{k}\neq k}$ 
is also determined by \eqref{3.15}. 
Notice that the message splitting in \eqref{message_spliting} results in 
total $KL$ message subsets and each message overlapped $K\+L\-1$ times over $KL$ message subsets.  
Thus, adding up all the inequalities associated with $\ls \cW^{\ell k} \rs_{\ell\in\cL, k\in\cK}$ yields the total degrees
of freedom outer bound as
\beq
\d4\d4\d4 \Lambda_d \d4&\leq&\d4 KL \min\!\bigg(\! M, \frac{LN}{K\+L\-1},  \nonumber \\ 
          \d4& &\d4 \hspace{1.0cm} \frac{\max\!\lp KM, (L\-1)N \rp}{K\+L\-1}, \frac{\max\!\lp (L\-1)M,N \rp}{K\+L\-1} \! \bigg). \label{3.22}
\eeq
Combining two bounds in \eqref{3.24} and \eqref{3.22} and realizing that 
$\frac{KL}{K\+L\-1}LN \geq LN$
for $K,L\geq 1$ yield the outer bound result in \eqref{theorem_MC_MAC_outer_bound_equation}.
\begin{figure}[t]
    \centering
    \includegraphics[width=6.0cm, height=6.0cm]{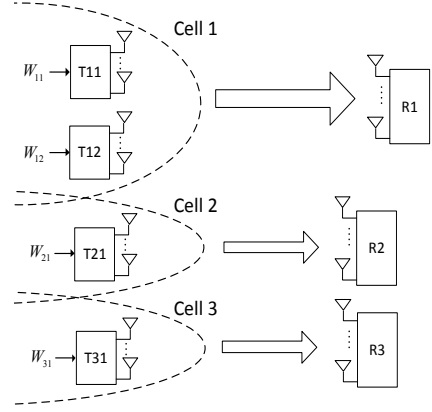}
    \caption{Heterogeneous network consisting of a $2$-user MIMO MAC (i.e., cell $1$) and 
             $2$-user MIMO interference channel (i.e., cell $2$ and $3$).}
    \label{Fig6}
\end{figure}
\end{proof}

In what follows, we will quote the result in this section to characterize the optimal degrees of freedom for 
two-cell and $K$-user MIMO MAC.  
\section{Achievability and Optimal Degrees of Freedom for Two-cell and $K$-user MIMO MAC} \label{section4}

Our base line algorithm is to explore the feasibility of the linear scheme utilizing the spatial dimensions under 
zero interference constraints. 
The achievable schemes utilize linear precoder at the transmitter and linear 
postprocessing linear filter $\bP_m\in\C^{K\beta\times N}$ at the receiver $m$ to generate $\beta$ interference 
free dimensions for each of users. 
The required antenna dimensions $M$ and $N$ for achieving the optimal degrees of freedom are found as a linear 
function of $K$ and the number of transmit streams. 

\begin{theorem} \label{theorem_MC_MAC_optimal_dof}
The two-cell and $K$-user MIMO MAC with nondegenerate channels, where the 
user and base station have $M\=K\beta$ and $N\=K\beta\+\beta$ antennas or $M\=K\beta\+\beta$ 
and $N\=K\beta$ antennas, respectively, has the optimal degrees of freedom of $2K\beta$ where $\beta$ is positive integer. 
\end{theorem}

\subsection{Converse of Theorem \ref{theorem_MC_MAC_optimal_dof}} \label{inner_bound_MC_MAC_two_cell}

When $L=2$, the outer bound in  \eqref{theorem_MC_MAC_outer_bound_equation} yields
\beq
\Lambda_d \!\leq\! 2\min\!\!\lp\! KM,\! N,\! \frac{K\!\max(\!KM,N\!)}{K\+1},\! \frac{K\!\max(\!M,N\!)}{K\+1} \!\rp. \label{4.0}
\eeq
Plugging $M\=K\beta\+\beta$ and $N\=K\beta$ in \eqref{4.0} returns
\beq
2\min\!\lp\! K(K\+1)\beta, K\beta, \frac{K^2 (K\+1)\beta}{K\+1}, K\beta \!\rp = 2K\beta. \label{4.0.1}
\eeq
When $M\=K\beta$ and $N\=K\beta\+\beta$, we have
\beq
2\min\!\lp\! K^2\beta, (K\+1)\beta, \frac{2K^3}{K\+1}, K\beta\!\rp = 2K\beta. \label{4.0.2}
\eeq
Combining two bounds in \eqref{4.0.1} and \eqref{4.0.2} verifies the converse. 
\subsection{Achievability of Theorem \ref{theorem_MC_MAC_optimal_dof}} \label{inner_bound_MC_MAC_two_cell}
The achievability is argued by showing that $\beta$ interference free dimensions per user are resolvable
by constructing achievable linear schemes. 
 
Independently encoded $\beta$ streams are transmitted as 
$\bx_{mk}\= \bW_{mk}\bs_{mk}$, $m\in\cL$ and $k\in\cK$
from user $mk$ to base station $m$ where $\bs_{mk}\=\lS s_{mk,1} \ldots s_{mk,\beta} \rS^T$$\in$$\C^{\beta\times 1}$ 
is the symbol vector carrying message $W_{mk}$ and 
$\bW_{mk}$$\in$$\C^{M\times \beta}$ denotes a linear precoding matrix. The 
signal received at base station $m$ can then be written as
\beq
\by_m\=\sum_{k=1}^{K}\!\bH_{m,mk}\bW_{mk}\bs_{mk}\+\sum_{k=1}^{K}\bH_{m,\barm k}\bW_{\barm k}\bs_{\barm k}\+\bz_m.  \label{4.1}
\eeq
where $\barm$ is defined as $\barm$$=$$\cL\bsh m$ for $\cL$$=$$\{1,2\}$. 

When $M\= K\beta\+\beta$ and $N\=K\beta$, user $\barm k$ picks the precoding matrix 
$\bW_{\barm k}$ such that  
\beq
\span\lp \bW_{\barm k} \rp \subset null\lp \bH_{m,\barm k} \rp. \label{4.2} 
\eeq
Since $\bH_{m,\barm k}$$\in$$\C^{K\beta\times (K\beta\+\beta)}$ is drawn from i.i.d. continuous distribution,
$\bW_{\barm k}$$\in$$\C^{M\times \beta}$ with $\rank(\bW_{\barm k})\=\beta$ can be found almost surely  
such that $\eqref{4.2}$ for all $k\!\in\!\cK$.

Applying percoders $\ls\bW_{\barm k}\rs_{k\in\cK, \barm\in\cL}$ designed by \eqref{4.2} 
to \eqref{4.1} gives the received vector at base station $m$ as 
\beq
\by_m=\sum\limits_{k\in\cL}\bH_{m,mk}\bW_{mk}\bs_{mk}+\bz_m. \nonumber
\eeq
The decodability of $K\beta$ streams from $\by_m$ requires 
$\bG_m$$=$$\lS \bH_{m,m1}\bW_{m1} \ \cdots \ \bH_{m,mK}\bW_{mK} \rS $$\in$$\C^{K\beta\times K\beta}$ to be a full rank. 
Since $\bW_{mk}$ in \eqref{4.2} is based on $\bH_{\barm,m k}$, $\bW_{mk}$ is mutually independent of $\bH_{m,mk}$.
Then, by Lemma \ref{lemma_linear_independent_columns} in Appendix \ref{appendix_linear_independent_columns}, 
$\bH_{m,mk}\bW_{mk}$$\in$$\C^{K\beta\times \beta}$ is a full rank and spans $\beta$-dimensional space with probability one. 
Since $\ls \bH_{m,mk}\bW_{mk} \rs_{k\in\cK}$ are independently realized by continuous distributions and 
each $\bH_{m,mk}\bW_{mk}$ spans $\beta$-dimensional subspace, the aggregated channel $\bG_m$$\in$$\C^{K\beta\times K\beta}$ 
spans $K\beta$-dimensional space almost surely. This ensures achievability of $2K\beta$ degrees of freedom for 
two-cell MIMO MAC.

To argue the achievability for $M\=K\beta$ and $N\=K\beta\+\beta$, define an out-of-cell 
interference alignment plane at base station $m$ as $\bP_m$$\in$$\C^{K\beta\times (K\beta\+\beta)}$. 
Denote a projected out-of-cell interference channel at the base station $m$ 
as $\bP_{m}\bH_{m,\barm k}$$\in$$\C^{K\beta\times K\beta}$, $k \in \cK$. Transmitter 
$\barm k$ for $k \in \cK$ designs its precoder $\bW_{\barm k}$ 
such that
$\span\lp \bW_{\barm k} \rp$$ \subset $$null\lp \bP_{m}\bH_{m,\barm k} \rp$ with $\rank(\bW_{\barm k})$$=$$\beta$
whose necessary and sufficient condition is 
\beq
\dim\lp null\lp \bP_{m}\bH_{m,\barm k} \rp \rp= \beta, \  k\in\cK. \label{4.4} 
\eeq 
Since $\bP_m\bH_{m,\barm k}$ is $K\beta\times K\beta$, it is not straightforward to directly extract $\beta$-dimensional
null space from the effective channel $\bP_m\bH_{m,\barm k}$.  
However, we show in the following that extracting $\beta$-dimensional null space from $\bP_{m}\bH_{m,\barm k}\in\C^{K\beta\times K\beta}$
is possible by aligning the null spaces of the out-of-cell interference $\ls \bH_{m,\barm k} \rs_{k\in\cK}$ to the row space of $\bP_{m}$, 
which is referred to as \emph{null space interference alignment}.

Followed by Lemma \ref{lemma_equivalent_form} in Appendix \ref{appendix_equivalent_form}, \eqref{4.4} is restated as
\beq
\dim\lp ran\lp \bH_{m,\barm k} \rp \cap null\lp\bP_m\rp \rp \= \beta, \ k \in \cK \label{4.5}.
\eeq 
This formulation suggests a relevant interpretation that if a $\beta$-dimensional column subspace of $\bH_{m,\barm k}$ lies in 
$null\lp \bP_m \rp$ or equivalently, if the $\beta$-dimensional row subspace of $\bP_{m}$ lies in 
$null\big( \bH_{m,\barm k}^* \big)$ for all $ k\in\cK$, \eqref{4.5} is conveniently accomplished. 
Thus, the feasible $\bP_m$ is a matrix whose row subspace has $\beta$-dimensional intersection subspace 
with the null space of $\{ \bH_{m,\barm k}^* \}_{k\in\cK}$.     
In what follows the feasibility of \eqref{4.5} 
is established by aligning $K\beta$ dimensional out-of-cell interference space to 
$(K-1)\beta$ dimensional subspace by using null space interference alignment.

Suppose a set of matrices $\{ \bH_{m,\barm k}^* \}_{k\in\cK}$ and corresponding null space basis  
$\ls \bN_{m,\barm k}\rs_{k\in\cK}$ where $\bN_{m,\barm k}$$\in$$\C^{(K\beta\+\beta)\times \beta}$.
To enable \eqref{4.5}, $\bP_m$$\in$$\C^{K\beta\times N}$ 
is formed by mapping $\beta$ columns of $\bN_{m,\barm k}$ to the $(k-1)\beta+1$th 
to $k\beta$th rows of $\bP_m$, i.e., $\bP_m$ is constructed by 
\beq
\bP_m \= \lS \bN_{m,\barm 1} \  \bN_{m,\barm 2} \ \cdots \bN_{m,\barm K} \rS^*. \label{4.6}
\eeq
Note that the construction in \eqref{4.6} with $\{\bN_{m,\barm k}\}_{k\in\cK}$
always ensures $\rank(\bP_m)\=K\beta$ and  
$\dim\lp null\lp \bP_{m}\bH_{m,\barm k} \rp \rp= \beta$, $k\in\cK$, $m\in\cL$.
The mapping from columns of $\bN_{m,\barm k}$ to rows of $\bP_m$ is not unique.
In fact, since the condition in \eqref{4.5} describes the required condition about the right matrix null space of $\bP_m$, 
multiplying any full rank matrix $\bPi\in\C^{K\beta\times K\beta}$ to the left side of $\bP_m$ does not change the dimension 
condition in \eqref{4.4}, i.e.,   
\beq
\dim\!\lp null\!\lp \bPi\bP_{m}\bH_{m,\barm k} \rp \!\rp\=\dim\!\lp null\!\lp \bP_{m}\bH_{m,\barm k} \rp\! \rp\= \beta,  k\in\cK. \nonumber
\eeq

Given $\ls \bP_m \rs_{m\in\cL}$ in \eqref{4.6}, we find $\bW_{\barm k}$ such that
$\span\lp \bW_{\barm k} \rp$$\subset$$ null\lp \bP_{m}\bH_{m,\barm k} \rp$ for $k\in\cK, \barm\in\cL$. 
Then, the projected channel output at the base station $m$ is given by
\beq
\bP_m\by_m \= \sum_{k=1}^{K}\bP_m\bH_{m,mk}\!\bW_{mk}\bs_{mk}\+\bP_m\bz_m \= \bP_m\bG_m \tbs_m \+ \tbz_m \nonumber
\eeq 
where $\bG_m$$=$$[ \bH_{m,m1}\!\bW_{m1}  \cdots \bH_{m,mK}\!\bW_{mK} ]$, 
$\tbz_m$$=$$\bP_m\bz_m$, and $\tbs_m$$=$$[\bs_{m1}^T \cdots \bs_{mK}^T]^T$.    
For decodability, we need to check that $\bP_m\bG_m$ has linearly independent columns. 
Note that $\bP_m$ and $\bG_m$ are based on continuous distribution and mutually independent. 
Thus, Lemma \ref{lemma_linear_independent_columns} verifies that $\Pr\big( \det\big(  \bP_m\bG_m \big)\=0 \big)\=0$ 
implying the decodability of $K\beta$ interference free streams per cell. 
This ensures $2K\beta$ degrees of freedom for two cell MIMO MAC.

\section{Conclusions} \label{section_conlusions}
We have characterized the degrees of freedom region for the homogeneous $L$-cell and $K$-user MIMO MAC. 
We presented a degrees of freedom outer bound and linear achievable schemes for a few cases
that obtain the optimal degrees of freedom.
Transmit zero forcing is optimal in terms of the achievable degrees of freedom. 
The uplink scenario motivates us to build null space interference alignment 
scheme (with $N>M$) that promises the optimal degrees of freedom of $2K\beta$ for two cell case for arbitrary number of users.
By the uplink and downlink duality, the degrees of freedom results in this paper are also applicable to the 
downlink.


%
%
\appendices
\section{} \label{appendix_linear_independent_columns}

\begin{lemma} \label{lemma_linear_independent_columns}
Given $\bA$$\in$$\C^{m\times n}$ and $\bB$$\in$$\C^{n \times l}$ with $n\geq\max(m,l)$ where $\bA$ and $\bB$ 
are drawn from i.i.d. continuous distributions and are mutually independent, 
$\bA\bB$ has $\rank(\bA\bB)\=\min(m,l)$ with probability one.
\end{lemma}
\begin{proof}
First, we assume $\min(m,l)$$=$$m$ and decompose $\bB$$=$$ \big[ \hbB \ \bB{'} \big]$ where 
$\hbB\in\C^{n \times m}$ is formes by taking the first $m$ columns of $\bB$ and 
$\bB{'}\in\C^{n \times (l\-m)}$ is composed of columns from $m\+1$ to $l$  
of $\bB$. Then, about $\rank(\bA\bB)$ we have 
\beq
\rank( \bA\hbB )\leq \rank( \bA\bB\=[ \bA\hbB \ \bA\bB{'}] ) \leq \min(m,l)\=m \label{ap1.0}.
\eeq
Note that when $\min(m,l)\=l$, we need to consider the matrix $\bB^*\bA^*$
and it is handled similarly to the case $\min(m,l)\=m$. Thus, we omit the case $\min(m,l)\=l$ and 
focus on $\min(m,l)\=m$.

We further decompose $\bA$$=$$\big[ \bbA \ \tbA \big]$ and 
$\hbB^*$$=$$ \big[ \bbB \ \tbB \big]$ where $\bbA$$\in$$\C^{m\times m}$ and 
$\bbB$$\in$$\C^{m\times m}$ are leading principal submatrices 
of $\bA$ and $\hbB^*$, respectively, and $\tbA$$\in$$\C^{m\times (n\-m)}$ 
and $\tbB$$\in$$\C^{m\times (n\-m)}$ are submatrices corresponding to columns 
from $m\+1$ to $n$ of $\bA$ and $\hbB^*$, respectively. 

We claim $\Pr\!\big( \big| \!\det\!\big(\bA\hbB\big)\! \big|\! >\!0  \big)=1$. 
The claim is verified by providing the converse, i.e., $\Pr\!\big( \!\det\!\big( \bA\hbB \big)  \!\= \!0 \big)$$=$$0$. 
Since $\bA$ and $\hbB$ are drawn from i.i.d. continuous distributions, their principal 
submatrices $\bbA$ and $\bbB^*$ (which are square matrices) have 
$\rank\lp \bbA \rp$$=$$m$ and $\rank\lp \bbB^* \rp$$=$$m$ almost surely, respectively.
Now, we have
\beq
\!\Pr\! \big( \! \det\!\big( \!\bA\hbB \big)  \!\= 0 \big)
\!\!\d4&=&\d4\!\! \Pr\! \big( \det\big( \bbA\bbB^* + \tbA\tbB^* \big)  \= 0\big)  \nonumber \\
\!\!\d4&=&\d4\!\! \Pr\! \big(  \det\big( \bbA\bbB^* \big) \nonumber \\
\!\!\d4& &\d4\!\! \hspace{1.0cm}  \times \det\big( \bI_{m} \+ \big( \bbA\bbB^* \big)^{-1} \tbA\tbB^*  \big)  \=0 \big)  \nonumber \\
\!\!\d4&=&\d4\!\! \Pr\! \big(\!  \big\{ \! \det\!\big(\! \bbA\bbB^* \!\big)\=0 \big\}  \nonumber \\
\!\!\d4& &\d4\!\! \hspace{1.0cm}  \cup \big\{ \! \det\!\big( \bI_{m} \+ \big(\! \bbA\bbB^{\!*} \!\big)^{-1} \tbA\tbB^* \!\big)\=0 \big\} \big)  \nonumber.
\eeq    
By using the fact that both of $\bbA\bbB^*$ and 
$\bI_{m} \+ \lp \bbA\bbB^* \rp^{-1} \tbA\tbB^*$ are invertible $m \times m$ matrices, 
we obtain
\beq
\!\Pr\! \big( \! \det\!\big( \!\bA\hbB \big)  \= 0 \big)
\d4&\leq&\d4 \Pr\! \big(\!  \det\!\big( \bbA\bbB^* \big)\=0 \big) \nonumber \\
     \d4& &\d4 \hspace{0.5cm}\+ \Pr\! \big(\!  \det\!\big( \bI_{m} \+ \lp\! \bbA\bbB^* \!\rp^{-1}\! \tbA\tbB^* \big)\=0 \big)  \nonumber
\eeq
where $\Pr\! \big(\!  \det\!\big(\! \bbA\bbB^* \big)\=0 \big)\=0$ and 
$\Pr\! \big(\!  \det\!\big( \bI_{m} \+ \lp\! \bbA\bbB^* \!\rp^{-1} \tbA\tbB^*\! \big)\=0 \big)\=0$.
Consequently, we get $\Pr\big( \det\big( \bA\hbB \big)  \= 0 \big)\=0$. This concludes the proof.  
\end{proof}

\section{} \label{appendix_equivalent_form}

\begin{lemma} \label{lemma_equivalent_form}
For any $\bP_m\in\C^{M \times N}$ and nondegenerate 
$\bH_{m, \barm k}\in\C^{N\times M}$ with $\rank\lp \bH_{m,\barm k} \rp=M$ and $N>M$, 
\beq 
\dim\!\lp null\!\lp \bP_{m}\bH_{m,\barm k} \!\rp \!\rp 
               \=\dim\!\lp ran\!\lp \bH_{m,\barm k} \!\rp \cap null\!\lp \bP_m \!\rp \!\rp. \label{lemma_equivalent_form_equation}
\eeq
\end{lemma}
\begin{proof}
By definition, $\dim\!\lp null\!\lp \bP_{m}\bH_{m,\barm k} \rp \rp$ is rewritten by
\beq
\d4 \dim \! \big( \!\!\!\d4&\{&\d4\!\!\! \ba \!\in\! \C^{M \times 1} \! : \! \bP_m\bH_{m,\barm k}\ba\=\bzero \} \big) \nonumber \\
\d4&=&\d4 \dim\!\big(\! \big\{ \ba \!\in\! \C^{M \times 1}\!: \!\bH_{m,\barm k}\ba \in null\lp\bP_m\rp \big\} \big) \label{4.3.1} \\
\d4&=&\d4 \dim\!\big(\! \big\{ \bb \!\in\! \C^{N \times 1}\!: \!\bb \in ran\!\lp \bH_{m,\barm k} \!\rp \nonumber \\ 
\d4& &\d4   \hspace{4cm} \& \bb\in null\!\lp \bP_m \!\rp\! \big\}\! \big) \label{4.3.2} 
\eeq
where \eqref{4.3.1} follows from the facts that $null\lp \bH_{m,\barm k} \rp=\phi$. 
In \eqref{4.3.2}, we use the fact that 
the mapping from $\ba$ to $\bb$ via $\bH_{m,\barm k}$ (i.e., $\bH_{m,\barm k}\ba=\bb$) for $\forall\ba\in\C^{M\times 1}$ 
is one-to-one if and only if $N\geq M=\rank\lp\bH_{m,\barm k}\rp$. 
Now the expression in \eqref{4.3.2} implies \eqref{lemma_equivalent_form_equation}. 
\end{proof}

\bibliographystyle{IEEEtran}
\bibliography{IEEEabrv,Globecom2011_rev3}

\end{document}